\documentclass[12, reqno]{amsart}
\usepackage{amsmath}
\usepackage{amscd,amsthm,amsfonts,amsopn,amssymb,verbatim}
\usepackage{mathrsfs}
 \usepackage{graphicx}
\numberwithin{equation}{section}

\newtheorem{theorem}{Theorem} %[section]

\newtheorem{corollary}[theorem]{Corollary}

\newtheorem{lemma}[theorem]{Lemma}

\theoremstyle{remark}

\def\be{\begin{equation}}
\def\ee{\end{equation}}

\def\ve{\varepsilon}
\def\Spec{\text {\rm Spec\,}}

\allowdisplaybreaks

\begin{document}
\setlength{\parskip}{2pt}

\title[Spectral aspects of the skew-shift operator]
{Spectral aspects of the skew-shift operator: \qquad A numerical perspective}
\author{Eric Bourgain-Chang}
\address
{Mechanical Engineering Department, University of California, Berkeley, CA 94720}
\email{ebc@berkeley.edu}
\maketitle

{\bf Abstract}
In this paper we perform a numerical study of the spectra, eigenstates, and Lyapunov exponents of the skew-shift counterpart to Harper's equation. This study is motivated by various conjectures on the spectral theory of these 'pseudo-random' models, which are reviewed in detail in the initial sections of the paper. The numerics carried out at different scales are within agreement with the conjectures and show a striking difference compared with  the spectral features of the Almost Mathieu model. In particular our numerics establish a small upper bound on the gaps in the spectrum (conjectured to be absent).


\section
{Introduction}

The almost Mathieu operator is the self-adjoint operator acting on $\ell^2(\mathbb Z)$ defined by
\be\label{1.1}
[H_{\lambda, \omega, \theta} u] (n) =2\lambda \cos \big(2\pi (n\omega+\theta)\big) u_n+ u_{n+1} +u_{n-1}
\ee
where $\omega, \theta\in \mathbb T=\mathbb R/\mathbb Z$ and $\lambda$ is the coupling parameter.
We always assume that $\omega$ is irrational (and later on, it is diophantine), so the potential in \eqref{1.1} is almost-periodic. A discussion of the physical motivation and background of \eqref{1.1} may be found in \cite{L} for instance. Let us recall that $H_{\lambda, \omega, \theta}$ is a model for the Hamiltonian of an electron in a one-dimensional lattice, subject to a potential, and is also related to the Hamiltonian of an electron in a two-dimensional lattice subject to a perpendicular magnetic field. Such models go back to the work of Peierls \cite{P}
related to the theory of the quantum Hall effect where \eqref{1.1} describes a Bloch electron in a magnetic field.
The case $\lambda=1$ is particularly important and called Harper's equation.
From the mathematical side, \eqref{1.1} has been extensively studied over the recent decades and the complete
spectral theory in the different regimes is understood (the main features will be summarized later).
An important property of \eqref{1.1} is the so-called Aubry duality relating the regimes $|\lambda|<1$ and $|\lambda|>1$.
A major break-through in this area came with the work of Jitomirskaya \cite{J}.

To be noted is that the theory of almost periodic Schr\"odinger operator has been developed in far greater generality than \eqref{1.1} (many
of its specific properties are in some sense `special') but we will not further elaborate on this here.
Our interest goes to the formally related Hamiltonian
\be\label{1.2}
[Hu](n) = 2\lambda\big(\cos(2\pi n^2 \omega)\big) u_n +u_{n+1}+ u_{n-1}
\ee
which we refer to as the skew shift $\rm{Schr\ddot{o}dinger}$ operators, since its potential can be generated from the orbits of the skew shift $T_\omega$ acting on
$\mathbb T^2$, defined by $T_\omega(x, y) =(x+y, y+\omega)$.
Such models are relevant to the theory of the quantum-kicked rotor, the quantum version of the classical Chirikov standard map.
More generally, Jacobi matrices on $\mathbb Z_+$ of the form
\be\label{1.3}
[Hu](n)= 2\lambda \big(\cos(n^\beta)\big) u_n+u_{n+1}+u_{n-1}
\ee
with $\beta >1$ were considered in the works of Griniasty-Fishman \cite{G-F} and Brenner-Fishman \cite{B-F} (written in the form \eqref{1.3},
assuming $\beta$ not an integer). As these authors point out, there is a large variety of potentials that are neither periodic nor incommensurate and the study of deterministric `pseudo-random' systems is of broad interest beyond quantum chaos, including to theoretical computer science. It is believed that the spectral theory of \eqref{1.2} and \eqref{1.3}, at least for $\beta>2$, resembles
that of $\rm{Schr\ddot{o}dinger}$ operator with a random potential, even at small disorder $\lambda$.
They were proposed in \cite{G-F}, \cite{B-F} as `pseudo-random' models and discussed heuristically.
So far, the main rigorous results for \eqref{1.2} assume $|\lambda|$ large.
See in particular H.~Kr\"uger's paper \cite{K1} and the related references.
For $0<|\lambda|\leq 1$, there are only a few contributions and they relate either to somewhat atypical frequencies $\omega$ or modifications of \eqref{1.2}
(see \cite{K2} and the discussion in that paper).
Thus at this point, there is no satisfactory mathematical theory that explains the expected phenomenology (which will be stated explicitly in \S3).

Taking $\omega=\frac{1+\sqrt 5}2$ the golden mean ratio, the purpose of this Note is to carry out some numerics related to the spectrum and
eigenstates of the operator
\be\label{1.4}
[Hu](n) =2(\cos 2\pi n^2\omega) u_n+u_{n+1} + u_{n-1}
\ee
which is the skew shift counterpart of Harper's equation.
The interest of such numerics is two-fold.
Firstly, it gives some understanding how, on a finite scale, the eigenvalue distribution and eigenvector localization compares
with the Harper case
\be\label{1.5}
[Hu](n)=2(\cos 2\pi n\omega) u_n+ u_{n+1}+ u_{n-1}.
\ee
Secondly, using numerics, one can in fact prove certain spectral properties of the full operator.
We illustrate this with the modest example of an upperbound on the size of possible gaps (if any) in the spectrum of \eqref{1.4}.
Obviously larger scale numerics would likely lead to better estimates.
We also indicate how, `in principle', appropriate numerics may permit one to start the multi-scale analysis underlying the approach for large $\lambda$,
in order to prove small $\lambda$ results (for instance in the context of \cite{K1}).

Returning to \eqref{1.3} and more generally operators of the form
\be\label{1.6}
[Hu](n)= f(n^\beta) u_n+u_{n+1} +u_{n-1}
\ee
with $f$ a nonconstant, continuous periodic function on $\mathbb R$, the case $\beta>0, \beta\not\in \mathbb Z$ may be studied by different methods.
It was shown for instance in \cite{K3} that for $H$ as in \eqref{1.6},
\be\label{1.7}
\Spec H = [\min f-2, \max f+2].
\ee

The paper is organized as follows. In the next section, we briefly review several aspects of the theory of 1D lattice $\rm{Schr\ddot{o}dinger}$ operator with various types of potentials, the almost Mathieu operator and also the conjectured behavior of models with weakly mixing potentials (for instance governed by skew-shift dynamics). In Section 3, we explain how information on certain spectral parameters (such as location and gaps in the spectrum) for the full Hamiltonian $H$ may be derived from properties of the eigenvalues and eigenvectors of `finite models' obtained by restriction of $H$ to a finite interval. In Section 4, these considerations are further specified in the context of the skew shift potential, which is our main interest in this work. The numerical implementation appears in Section 5. We compare the Harper model (the Almost Mathieu operator at the critical coupling $\lambda =1$) with its skew shift counterpart, with emphasize on their spectra, eigenfunction localization and Lyapunov exponents through the
energy range. Section 6 summarizes the conclusions and stresses some further research perspectives.

\bigskip

\section{Some background on discrete Schr\"odinger operators}

We only discuss the 1D model (on the lattice $\mathbb Z$).
Thus $\mathcal H =\ell^2(\mathbb Z)$ is the underlying Hilbert space.
We consider self-adjoint operators of the form
\be\label{2.1}
H=V+\Delta
\ee
where $V$ is a diagonal operator given by a bounded potential, i.e.
\be\label{2.2}
V=\sum_{n\in\mathbb Z} V_n (e_n\otimes e_n)
\ee
with $\{e_n\}$ the unit vectors of $\mathcal H$, $\Delta =$ lattice Laplacian on $\mathbb Z$, given by
\be\label{2.3}
\Delta =\sum_{n\in\mathbb Z} (e_n \otimes e_{n+1} + e_{n+1} \otimes e_n).
\ee
The potential is often introduced by evaluation of a given bounded function \hfill\break
$f:\Omega\to \mathbb R$ on the orbits of some dynamical system
$(\Omega, \mu, T)$ with $T$ an ergodic measure preserving transformation. Thus
\be\label{2.4}
V_n=f(T^n x)
\ee
with $x\in\Omega$ some base point.
Recall that if we start from a general bounded sequence $(V_n)_{n\in \mathbb Z}$, the format \eqref{2.4} may always be attained
by a classical construction which consists in taking for $\Omega$ the pointwise closure in $\mathbb R^{\mathbb Z}$ of the orbit of
$(V_n)_{n\in\mathbb Z}$ under the left shift $T$
\be\label{2.5}
(Tx)_n =x_{n+1} \ \text { for } \ x=(x_n)_{n\in\mathbb Z}
\ee
and for $\mu$ a Banach invariant mean (ergodicity is then achieved by decomposing the resulting dynamical system $(\Omega, \mu, T)$ in its
ergodic components).

Returning to \eqref{1.2}, there are two models extensively studied by many authors.
The first is the case of a random potential
\be\label{2.6}
V_\omega =\big(V_n(\omega)\big)_{n\in \mathbb Z}
\ee
with $V_n$ independent identically distributed random variables.
This is the 1D Anderson model which describes for instance transport in inhomogenous media (such as a metal alloy).
In this situation, $H=H_\omega$ and its main spectral features are the following (assuming the distribution non-constant).

\begin{itemize}
\item [(2.7)] \ For almost all $\omega$, the spectrum $\Spec (H_\omega)$ is the set
$$
\sum = [-2, 2] +\text { range } V
$$

\item[(2.8)] \ For almost all $\omega$, $H_\omega$ has pure point spectrum and the corresponding eigenfunctions are exponentially localized i.e. decay
exponentially for $|n|\to\infty$.
\end{itemize}

\stepcounter{equation}
Property (2.8) is referred to as "Anderson localization".
Let us emphasize that this discussion is 1D, which is a well understood theory (unlike in dimension $\geq 2$ where different
phenomena are expected).
There are many reference works, see for instance \cite{F-P}.
It is convenient to replace $V$ by $\lambda V$, where $\lambda\not= 0$ is called the disorder parameter.
Thus (2.8) remains true, also for small $\lambda$.

Another extensively studied model is \eqref{1.1}, i.e. the almost Mathieu operator with almost periodic potential
\stepcounter{equation}

\be\label{2.9}
V_n =2\lambda\cos \big(2\pi(n\omega+\theta)\big).
\ee
There is a long list of contributors and contributions to this subject, but we restrict ourselves to a summary of the final
state of the theory.
We always assume $\omega$ irrational, implying ergodicity of the potential and $\Spec H_{\lambda, \omega, \theta}$ independent of $\theta$
(as a set).
Say that $\omega$ is diophantine if it satisfies for some constant $C>0$
\be\label{2.10}
\Vert k\omega\Vert> \frac 1C |k|^{-C} \text { for all } k\in\mathbb Z\backslash \{0\}.
\ee

\stepcounter{equation}\label{2.11}
\begin{itemize}
\item[(2.11)] \ For $\lambda \not= 0$, $\Spec H_{\lambda, \omega}$ is a Cantor set of Lebesque measure
\be\label{2.12}
|\Spec H_{\lambda, \omega}|= 4|1-\lambda|
\ee
assuming moreover $\omega$ diophantine.

\item[(2.13)] \ For $0\leq |\lambda| <1, H_{\lambda, \omega}$ has purely absolutely continuous (ac) spectrum.

\stepcounter{equation}
\item[(2.14)] \ For $|\lambda|>1$, for almost all $\theta$, $H_{\lambda, \omega, \theta}$ has pure point spectrum and exhibits Anderson localization.

\stepcounter{equation}
\item[(2.15)] \ For $|\lambda|=1$ and almost all $\theta$, $H_{\lambda, \omega, \theta}$ has purely singular continuous spectrum.
\end{itemize}

See \cite{A-J}, \cite {J-K}, \cite{J}, \cite{G-J-L-S} for statements (2.11), \eqref{2.12}, (2.14), (2.15) respectively; (2.13) is a consequence of Aubry duality.
Some of the above results were actually proven in a stronger form, but this is not important for what follows.
Note also that, by \eqref{2.12}, the Harper operator has spectrum of zero Lebesgue measure which forces it automatically to
be Cantor.

Let us consider next the skew shift Schr\"odinger operator obtained by taking in \eqref{2.4} for $T$ the skew shift acting on the $2$-torus
$\mathbb T^2$, i.e.
$$
T(x, y)= (x+y, y+\omega)
$$
and $f=\lambda\cos \theta$ acting on the first coordinate.
Thus
\stepcounter{equation}
\be\label{2.16}
H_{\lambda,\omega, x, y} =2 \lambda V+\Delta \ \text { with } \ V_n =\cos \Big(x+ny+\frac {n(n-1)}2 \omega\Big)
\ee
and we assume again $\omega$ diophantine.

It is conjectured that \eqref{2.16} displays a spectral behaviour similar to the random case \eqref{2.6}, also for small $\lambda$.
In particular that $\Spec H_{\lambda, \omega, x, y}$ has no gaps and for most $x, y, H_{\lambda, \omega, x, y}$ has pure point spectrum
with Anderson localization.
This last statement has been proven for $|\lambda|$ sufficiently large and $(\omega, x, y) \in \mathbb T^3$ taken in a suitable set of positive
measure.
 From this respect, there is no difference with the almost periodic case; in particular \eqref{1.1} satisfies (2.14).
On the other hand, it was shown more recently in \cite{K1}, again for $\lambda$ sufficiently large (and $\omega$ satisfying
the stronger DC $\Vert k\omega\Vert>\frac c{k^2}$), that $\Spec H_{\lambda, \omega, x, y}$ contains at least an interval,
hence is not a Cantor set.

Returning to \eqref{2.1}, the 1D lattice $\rm{Schr\ddot{o}dinger}$ operator can be studied using the transfer matrix formalism (which is a powerful tool
specific to the 1D situation).
Given an arbitrary sequence $(u_n)_{n\in\mathbb Z}$, the equation
$$
Hu =Eu
$$
is equivalent with
\be\label{2.17}
\begin{pmatrix} u_{n+1} \\ u_n\end{pmatrix} =M_n (E) \begin{pmatrix} u_1\\ u_0\end{pmatrix}
\ee
where the transfer matrix $M_n(E)\in SL_2(\mathbb R)$ is given by
\be\label{2.18}
M_n(E)=\prod^1_{j=n} \begin{pmatrix} V_j-E&-1\\ 1&0\end{pmatrix}.
\ee
Assuming $V_n$  given by \eqref{2.4}, define
\be\label{2.19}
L_n(E) =\frac 1n\int \log \Vert M_n(E, x)\Vert dx
\ee
and the Lyapunov exponent
\be\label{2.20}
L(E) =\lim_{n\to\infty} L_n(E).
\ee
For ergodic $T$, one has
\be\label{2.21}
L(E) =\lim_{n\to \infty} \frac 1n\log \Vert M_n(E, x)\Vert \qquad x  \ a.s.
\ee

The Lyapunov exponent plays a prominent role in the spectral theory of $H$.
Recall in particular Kotani's theorem, stating the following:

Let $(a, b)\subset \mathbb R$ be an interval.
Then $H$ has no ac spectrum in $(a, b)$, i.e.
\be\label{2.22}
\sum_{ac} (H) \cap (a, b) =\phi
\ee
if and only if
\be\label{2.23}
L(E)>0 \text { for almost all energies $E\in (a, b)$}
\ee
(again assuming the underlying transformation $T$ ergodic).

For both \eqref{1.1} and \eqref{2.16}, the Lyapunov exponent satisfies
\be\label{2.24}
L(E) \geq\max (0, \log|\lambda|).
\ee
Moreover, for the almost Mathieu operator, (2.14) becomes an equality if \hfill\break
$E\in \Spec H$, which in view of Katani's theorem is consistent
with (2.13)--(2.15).
Note however that while positivity of the Lyapunov exponent excludes ac spectrum, it does not necessarily imply point spectrum.
For random potentials, the Lyapunov exponents are positive, also at small disorder.
The same is conjectured to be true for the skew-shift potential \eqref{2.16}, which is a central issue in this discussion.
Establishing positivity of the Lyapunov exponents of \eqref{2.16} (at a given $\lambda$) immediately implies absence of ac spectrum.
However, it is also the first step in a multi-scale analysis leading to Anderson localization, so far only proven for large
$\lambda$.
Going one step further, sufficient information about the function
\be\label{2.25}
\frac 1n \log \Vert M_n (E; x, y)\Vert
\ee
of the three variables $(E, x, y)$ at a sufficiently large scale $n$ would already suffice for this purpose.
In a similar vein, it would permit an extension of the theory developed in \cite{K1} to other values of $\lambda$, proving that
$\Spec H_{\lambda, \omega, x, y}$ is not a Cantor set.
While in principle a computer assisted approach to some of the conjectures above may be possible because it only involves the analysis at some fixed scale, the technical difficulties are considerable.
First, one would have to determine an appropriate formulation of the bootstrap process and the ncessary numerical input at some fixed cale $n$.
This bootstrap argument depends on a rather lengthy and sophisticated mathematical analysis that would have to be improved.
Also, the required scale $n$ may be computationally infeasible.
On a much more modest level, it turns out that the numerical evaluation of $L_n(E)$ for large $n$ is already a nontrivial
task, due to the accumulation of errors in calculating the matrix products.
Such numerics were carried out several years ago by W.~Schlag (private communication) and turned out to be inconclusive at large scale for the above reason.

A few more comments on `pseudo-random' potentials:
while the skew shift $T$ discussed above is an example of a weakly mixing potential, deterministic strongly mixing potentials may be obtained, considering for instance the doubling map $x\mapsto 2x$ acting on $\mathbb T$ or a hyperbolic toral automorphism $A\in SL_2(\mathbb Z)$
acting on $\mathbb T^2$.

Thus one would define
\be\label{2.26}
V_n=2\lambda \cos \pi 2^n\omega
\ee
in the first case and
\be\label{2.27}
V_n=2\lambda \cos \Big\langle A^n \begin{pmatrix} x\\ y\end{pmatrix}, e_1\Big\rangle
\ee
in the second setting.
For these models, a closer analogy with the random $\rm{Schr\ddot{o}dinger}$ operator theory may be proven.
In particular, at small $\lambda$, one obtains positivity of the Lyapunov exponents and Anderson localization.
However, the underlying techniques do not seem applicable to skew shift dynamics.

For the operator \eqref{1.3}, \eqref{1.6} with $\beta>1$ and $\beta\not\in \mathbb Z$ (this last assumption is essential),
the semi-continuity methods from \cite {L-S} may be applied.
In addition to \eqref{1.7}, it was proven in \cite{K3} that for $H_{f, \beta}$ defined in \eqref{1.6} and $1\leq r<\beta<r+1,~r\in~\mathbb Z$,
\be\label{2.28}
\sum_{ac} (H_{f, \beta}) \subset\bigcap_{a_0, a_1, \ldots, a_r\in\mathbb R} \, \sum_{ac} (H_{a_0, a_1, \ldots, a_r})
\ee
where $H_{a_0, \ldots, a_r}$ is the $\rm{Schr\ddot{o}dinger}$ operator with potential $f\big(\sum^r_{j=0} a_j n^j)$.

The results from \cite{K4} have further implications on the Lyapunov exponents of \eqref{1.6}.
In particular, for \eqref{1.3}, i.e. $f(t)=2\lambda \cos t$, Lyapunov exponents do not vanish for all energies $E\not\in \mathcal E_\lambda$ , where
$\mathcal E_\lambda \subset \mathbb R$ is a set satisfying $|\mathcal E_\lambda|\to 0$ for $\lambda\to 0$.

\section
{Finite scale restrictions}

It is possible to derive spectral information for te full Hamiltonian $H$ by studying its restriction to finite boxes.

Given an interval $I$ in $\mathbb Z$, denote $H_I$ the finite matrix (indexed by $I$) defined by
$$
H_I(m, n)= H(m, n) \text { for } m, n\in I.
$$
Thus, if for instance $I=[0, N-1]$ and $H$ is given by \eqref{2.1},
$$
H_I= \left[\begin{matrix}
V_0&1&0&0&\cdots&0\\
1&V_1&1&0&\ldots&0\\
0&1&V_2&1&\ldots &0\\
0&0&1&V_3&&\vdots\\
\vdots&\vdots&0&1 &\ddots&1\\
\vdots&\vdots&\vdots&\vdots &&&\\
0&0&0&0&&1V_{n-1}\\
\end{matrix}\right]
$$
We will need a few  elementary results that relate $\Spec H$ with $\Spec H_I$ and can be used for numerical calculations.
They are well-known and we record them here in the explicit form needed.

\begin{lemma}\label{Lemma1}
Let $I=[a, a+N-1]\subset\mathbb Z$.

Let $E\in \mathbb R, \xi =\sum_{n\in I} \xi_n e_n, \Vert\xi\Vert=\big(\sum_{n\in I}|\xi_n|^2\Big)^{\frac 12} =1$ such that
\be\label{3.1}
\Vert H_I\xi -E\xi\Vert <\ve.
\ee
Then
\be\label{3.2}
\text{dist\,} (E, \Spec H)<\ve+|\xi_a|+|\xi_{a+N-1}|.
\ee
\end{lemma}

\begin{proof}

Define the vector $\eta \in \ell^2(\mathbb Z)$ by
$$
\begin{cases}
\eta_n=\xi_n \text { if } n\in I\\
\eta_n=0 \text { if } n\in\mathbb Z\backslash I.
\end{cases}
$$
Obviously $\Vert\eta\Vert =\Vert\xi\Vert= 1$.

We compute $H_\eta$
$$
(H\eta)_n=V_n \eta_n+\eta_{n+1} +\eta_{n-1}=
\begin{cases}
V_n\xi_n+\xi_{n+1}+\xi_{n-1} =(H_I\xi)_n \text { if } a<n<a+N-1\\
V_a \xi_a +\xi_{a+1} =(H_I\xi)_a \text { if } n=a\\
V_{a +N-1} \xi_{a+N-1} +\xi_{a+N-2} =(H_I\xi)_{a+N-1} \text { if } n=a+N-1\\ \xi_a \text { if } n=a-1\\
\xi_{a+N-1} \text { if } n=a+N\\
0 \ \text { if } n<a-1 \text { or }  n> a+N.
\end{cases}
$$
Hence
\be\label{3.3}
H\eta =\sum_{n\in I} (H_I \xi)_n e_n+\xi_a e_{a-1} +\xi_{a+N-1} \, e_{a+N}
\ee
and
\begin{align}\label{3.4}
H\eta-E\eta &= \sum_{n\in I} \big((H_I\xi)_n -E\xi_n\big) e_n+\xi_a e_{a-1}+\xi_{a+N-1} \, e_{a+N}\nonumber\\
\Vert H\eta -E\eta\Vert &\leq \Vert H_I\xi -E\xi\Vert +|\xi_a|+|\xi_{a+N-1}|\nonumber\\
&< \ve +|\xi_a|+|\xi_{a+N-1}|.
\end{align}
\end{proof}

Conclusion \eqref{3.2} then follows from \eqref{3.4} and the spectral theorem.

As an immediate consequence, we get

\begin{corollary}\label{Corollary2}

{}{}{} {} {} \qquad

Let $I=[a, a+N-1] \subset \mathbb Z$.

Let $\xi^{(1)}, \ldots, \xi^{(N)}$ be the normalized eigenvectors of $H_I$ and $\lambda_1, \ldots, \lambda_N$ the corresponding eigenvalues.

Then, for each $j=1, \ldots , N$
\be\label{3.5}
\text{\rm dist\,} (\lambda_j, {\rm \Spec} H)\leq |\xi_a^{(j)}|+|\xi_{a+N-1}^{(j)}|
\ee
\end{corollary}

This property is clearly of interest in order to establish an upperbound on possible gaps in the spectrum of $H$, by considering eigenvalues
and eigenvectors of a restriction $H_I$.
In view of \eqref{3.5}, only those eigenvalues of $H_I$ are of interest for which the corresponding eigenvector is small at the edges
of the box $I$.
Recall also that for the skew-shift $\rm{Schr\ddot{o}dinger}$ operator, one expects strong localization of the eigenvectors so the boundary
contribution for most of them should be quite small.

Next, we prove in some sense a converse property.

\begin{lemma}\label{Lemma3}

{}{}{}\qquad

Let $N\geq 2$ be a positive integer.
Let $0<\ve<1$ be arbitrary.

If $E\in \text{\rm Spec\,} H$, then there is an interval $I\subset \mathbb Z$ of size $N$ and an eigenvalue $E'$ of $H_I$ such that
\be\label{3.6}
|E-E'|<\sqrt \frac2N +\ve.
\ee
\end{lemma}

\begin{proof}
Since $E\in \Spec H$, there is a finitely supported vector $\eta =\sum' \eta_n e_n$, $\Vert\eta\Vert=1$ such that
\be\label{3.7}
\Vert H\eta -E\eta\Vert <\frac \ve 4.
\ee
Let $I=[a, a+N-1]$, with $a\in\mathbb Z$ to be specified, and define the vector \hfill\break $\eta'= \sum_{n\in I} \eta_n e_n$. Then
$$
(H_I\eta')_n=\begin{cases}
V_n \eta_n+ n_{n-1}+\eta_{n+1} =(H\eta)_n \text { if } a<n<a+N-1\\
V_a \eta_a +\eta_{a+1} =(H\eta)_a -\eta_{a-1} \text { if } n=a\\
V_{a+N-1} \eta_{a+N-1} +\eta_{a+N-2} =(H\eta)_{a+N-1} -\eta_{a+N} \text { if } n=a+N-1.
\end{cases}
$$
Hence
$$
H_I\eta'-E\eta'=\sum^{a+N-2}_{n=a+1} (H\eta-E\eta)_n \, e_n +\big(((H-E)\eta)_a -\eta_{a-1}\big) e_a +
\big(((H-E)\eta)_{a+N-1}-\eta_{a+N}\big) e_{a+N-1}.
$$
Therefore
$$\Vert H_I\eta'-E\eta'\Vert^2_2=\ell_1+\ell_2+\ell_3,$$where
\begin{align}
\ell_1&=\sum^{a+N-2}_{n=a+1} |\big((H-E)\eta\big)_n|^2\\
\ell_2&=|\big((H-E)\eta\big)_a -\eta_{a-1}|^2\\
\ell_3&=|\big((H-E)\eta\big)_{a+N-1} -\eta_{a+N}|^2.
\end{align}
We claim that we can choose $a\in \mathbb Z$ such that
\be\label{3.11}
\ell_1+\ell_2+\ell_3<\Big(\frac {2+\ve}N+\frac {\ve^2}{16}\Big) \Big[\sum^{a+N-1}_{n=a} |\eta_n|^2\Big] =
\Big(\frac {2+\ve} N+\frac {\ve^2}{16}\Big) \Vert\eta'\Vert^2_2.
\ee

To prove that there exists $a\in \mathbb Z$ such that \eqref{3.11} holds, simply sum both  sides of \eqref{3.11} over $a\in\mathbb Z$
and show that the left hand side is smaller than the right hand side.

Thus
$$
\begin{aligned}
\sum_{a\in \mathbb Z} \ell_1 =\sum_{a\in\mathbb Z} \, \sum^{N-2}_{n=1} |\big((H-E)\eta)_{a+n}|^2&=\\
 (N-2) \sum_{n\in\mathbb Z} |\big((H-E)\eta\big)_n|^2 &=(N-2) \Vert H\eta-E\eta\Vert^2\\
&<(N-2)\frac {\ve^2}{16}
\end{aligned}
$$
by \eqref{3.7}.

Next
$$
\begin{aligned}
\sum_{a\in\mathbb Z} \ell_2&= \sum_{a\in\mathbb Z}|\big((H-E)\eta\big)_a|^2\\
&+\sum_{a\in\mathbb Z} |\eta_{a-1}|^2\\
&-2\text{\,Re\,} \Big[\sum_{a\in \mathbb Z} \big((H-E)\eta\big)_a \, \bar\eta_{a-1}\Big]\\
&\leq \Vert (H-E)\eta\Vert^2+\Vert\eta\Vert^2 +2\sum_{a\in\mathbb Z} |\big((H-E)\eta\big)_a | \, |\eta_{a-1}|\\
&<\frac {\ve^2}{16} +1+2\sum_{a\in\mathbb Z}|\big((H-E)\eta\big)_a| \ |\eta_{a-1}|.
\end{aligned}
$$
By the Cauchy-Schwarz inequality
$$
\begin{aligned}
\sum_{a\in \mathbb Z} |\big((H-E)\eta\big)_a | \, |\eta_{a-1}|&\leq \Big(\sum_{a\in\mathbb Z}|\big(( H-E)\eta\big)_a|^2\Big)^{\frac 12}
\Big(\sum_{a\in\mathbb Z} |\eta_{a-1}|^2\Big)^{\frac 12}\\
&=\Vert(H-E)\eta\Vert . \Vert \eta\Vert_2<\frac \ve 4.
\end{aligned}
$$
Thus
$$
\sum_{a\in \mathbb Z} \ell_2 < 1+\frac {\ve^2}{16} +2.\frac \ve 4 =\Big(1+\frac\ve 4\Big)^2
$$
and similarly
$$
\sum_{a\in\mathbb Z}\ell_3 <\Big( 1+\frac\ve 4\Big)^2.
$$
Therefore
\begin{align}\label{3.12}
\sum_{a\in\mathbb Z} [\ell_1+\ell_2+\ell_3]& < (N-2) \frac {\ve^2}{16} +2\Big(1+\frac\ve 4\Big)^2\nonumber\\
& = N\frac {\ve^2}{16} +2\Big(1+\frac\ve 2\Big).
\end{align}
Summing the right hand side of \eqref{3.11} over $a\in\mathbb Z$, we obtain indeed
$$
\begin{aligned}
&\Big(\frac {2+\ve}N+\frac {\ve^2}{16}\Big) \sum_{a\in\mathbb Z} \Big(\sum^{a+N-1}_{n=a} |\eta_n|^2\Big) =\\
& \Big(\frac{2+\ve}N+\frac {\ve^2}{16}\Big) N\, \sum_{n\in\mathbb Z}|\eta_n|^2 =2+\ve+\frac {\ve^2}{16} N= \eqref{3.12}.
\end{aligned}
$$
This proves that we can find some $a\in\mathbb Z$ for which \eqref{3.11} holds, thus such that
\be\label{3.13}
\Vert H_I\eta'-E\eta'\Vert^2 < \Big(\frac {2+\ve}N +\frac {\ve^2}{16} \Big) \Vert\eta'\Vert^2_2.
\ee
Define
$$
\xi=\frac {\eta'}{\Vert\eta'\Vert_2}.
$$
which satisfies by \eqref{3.13}
$$
\Vert H_I\xi -E\xi\Vert <\Big(\frac {2+\ve}N+\frac {\ve^2}{16}\Big)^{\frac 12} <\sqrt{\frac 2N} +\ve.
$$
This means that there is an eigenvalue $E'$ of $H_I$ satisfying
$$
|E-E'|<\sqrt{\frac 2N}+\ve
$$
which Proves Lemma \ref{Lemma3}.
\end{proof}

The interest of Lemma \ref{Lemma3} is to establish conversely the presence of gaps in $\Spec H$.
The numerics in this paper also require evaluation of $\sup (\Spec H)$ and $\inf (\Spec H)$  for which we have a slightly
better estimate using spectra of finite restrictions.

\begin{lemma}\label{Lemma4}
Under the assumptions of Lemma \ref{Lemma3}, there is an interval $I\subset\mathbb Z$ of size $N$ and an eigenvalue $E'$ of $H_I$ such that
\be\label{3.14}
E'>\sup (\Spec H)-\frac 2N-\ve.
\ee
\end{lemma}

\begin{proof}
We proceed again by an averaging argument.

Denote $I_a= [a, a+N-1]$.
Clearly
\be \label{3.15}
\frac 1N\sum_{a\in\mathbb Z} 1_{I_a} (x) 1_{I_a}(y) =\begin{cases}
1-\frac{|x-y|}N \text { if } x, y\in\mathbb Z, |x-y|\leq N\\
0 \ \text { if } \ x, y\in\mathbb Z, |x-y|>N.
\end{cases}
\ee
Denote for $\eta\in\ell^2 (\mathbb Z)$ by $\eta_I$ the vector $\sum_{n\in I} \eta_n e_n$.

It follows from \eqref{3.15} that
$$
\frac 1N\sum_{a\in\mathbb Z} H_{I_a} =\frac 1N \sum_{a\in\mathbb Z} 1_{I_a} H 1_{I_a}= V+\Big(1-\frac 1N\Big) \Delta =H-\frac 1N\Delta.
$$
Take $\eta \in \ell^2(\mathbb Z), \Vert\eta\Vert_2=1$ such that $\langle H\eta, \eta\rangle > \sup (\Spec H)-\ve$.

It follows that
$$
\frac 1N\sum_{a\in\mathbb Z} \langle H_{I_a}\eta_{I_a}, \eta_{I_a}\rangle =\langle H\eta, \eta\rangle -\frac 2N Re \Big(\sum \eta_n \bar\eta_{n+1}\Big)
>\sup (\Spec H)-\ve-\frac 2N.
$$
Also
$$
\frac 1N\sum_{a\in\mathbb Z} \Vert\eta_{I_a}\Vert_2^2 =1.
$$
Hence, there is $a\in\mathbb Z$ such that
$$
\langle H_{I_a} \eta_{I_a}, \eta_{I_a}\rangle > \big(\sup (\Spec H)-\ve -\frac 2N\big) \Vert\eta_{I_a}\Vert^2_2
$$
and the claim in Lemma \ref{Lemma4} follows.

\begin{corollary}\label{Corollary5}
Denote
$$
\sigma_+ =\sup(\Spec H) \text { and } \sigma_- =\inf (\Spec H).
$$
Then
\be\label{3.16}
\sigma_+ \leq \sup_{|I|=N} \sup(\Spec H_I)+\frac 2N \text { and } \sigma_- \geq\inf_{|I|=N} \inf(\Spec H_I)-\frac 2N.
\ee
\end{corollary}

At first sight, the expressions on the r.h.s. of \eqref{3.14}, \eqref{3.16} seem useless for numerics because
they involve all intervals $I\subset \mathbb Z$ of size $N$.

In the situation of an ergodic Jacobi operator
\be\label{3.17}
H^{(x)} =V+\Delta \text { with } V_n=f(T^n x)
\ee
observe that if
$$
H_I^{(x)} =(H^{(x)}_{m, n})_{m, n\in I}
$$
then for $I=[a, a+N-1]$, clearly
\be\label{3.18}
H_I^{(x)} =H_{[0, N-1]}^{(T^a x)}.
\ee
Therefore, denoting $\lambda_+(x)$, resp $\lambda_- (x)$, the largest and smallest eigenvalues of $H^{(x)}_{[0, N-1]}$, we have
\be\label{3.19}
\sup_{|I|=N} \sup(\Spec H_I)=\sup_x\lambda_+ (x)
\ee
and
\be\label{3.20}
\inf_{|I|=N} \inf (\Spec H_I)=\inf_x \lambda_-(x).
\ee
Consequently, we obtain

\begin{corollary}\label{Corollary6}

Let $H^{(x)}$ be an ergodic Jacobi operator \eqref{3.17} and $\sigma_+, \sigma_-$ defined as in Corollary \ref{Corollary5}.
Let $N\geq 2$ be an integer and $\lambda_+(x)$, resp $\lambda_-(x)$ the largest and smallest eigenvalue of the $N\times N$-matrix
$H^{(x)}_{[0, N-1]}$.
Then
\be\label{3.21}
\sigma_+ \leq \max_x \lambda_+ (x)+{\frac 2N}
\ee
and
\be\label{3.22}
\sigma_-\geq \min_x \lambda_-(x) -{\frac 2N}.
\ee
\end{corollary}

In the particular case of the skew shift $\rm{Schr\ddot{o}dinger}$ operator
\be\label{3.23}
H=2\cos 2\pi n^2\omega+\Delta
\ee
the underlying dynamics is the skew shift $T$ on $\mathbb T^2$ mapping $(x, y)$ to $(x+y, y+2\omega)$. Thus we define
\be\label{3.24}
H^{(x, y)} = 2\sum_{n\in\mathbb Z} \cos 2\pi(n^2\omega+ny+x) e_n\otimes e_n+\Delta
\ee
and
\be\label{3.25}
H_N^{(x, y)} = H^{(x, y)}_{[0, N-1]} =2\sum^{N-1}_{n=0} \cos 2\pi (n^2\omega+ny+x) e_n\otimes e_n+\sum^{N-1}_{n=0}
(e_n\otimes e_{n+1} + e_{n+1} \otimes e_n).
\ee
Note the following property.
Assume $\xi =\sum_{n=0}^{N-1} \xi_n e_n $ satisfies
$$
H_N^{(x, y)} \xi=E\xi
$$
and let $\xi'= \sum_{n=0}^{N-1} (-1)^n \xi_n e_n$.
Then
$$
H_N^{(x+\frac 12, y)}\xi' =-E\xi'.
$$
Denoting
\be\label{3.26}
\lambda_+=\max_{0\leq x, y\leq 1}\lambda_+(x, y)
\ee
\be\label{3.27}
\lambda_- =\min_{0\leq x, y\leq1} \lambda_- (x, y)
\ee
with $\lambda_+(x, y)$ and $\lambda_-(x, y)$ the largest and smallest eigenvalue of $H_N^{(x, y)}$, it follows from the preceding
that $\lambda_- =-\lambda_+$. Hence

\begin{corollary}\label{Corollary7}
Let $H$ be as in \eqref{3.23}, $N\geq 2$ and $\lambda_+$ defined by \eqref{3.26}.
Then
\be\label{3.28}
\Spec H\subset \Big[-\lambda_+ - \frac 2N, \lambda_++\frac 2N\Big]
\ee
\end{corollary}
\end{proof}

\section
{Bounding spectral gaps for the skew-shift potential}

Recall the skew-shift $\rm{Schr\ddot{o}dinger}$ operators \eqref{1.4}
\be\label{4.1}
[Hu](n)=2 (\cos 2\pi n^2 \omega) u_n+u_{n+1} +u_{n-1} \text { with } \omega=\frac {1+\sqrt 5}2
\ee
and denote
$$
\sigma_+ =\sup (\Spec H) \quad \sigma_-=\inf (\Spec H)=-\sigma_+.
$$

We can make a numerical approximation of $\sigma_+$ by using Corollary \ref{Corollary6}.
Thus we choose a large $N$ and denote $\lambda_+(x, y)$ the largest eigenvalue of $H_N^{(x, y)}$ defined by \eqref{3.23}
\be
\label{4.2}
H_N^{(x, y)} =2\sum_{n=0}^{N-1} \cos 2\pi (n^2\omega+ny+x) e_n\otimes e_n+\sum^{N-1}_{n=0} (e_n\otimes e_{n+1} +e_{n+1} \otimes e_n).
\ee
According to \eqref{3.26}, one has
\be\label{4.3}
\sigma_+ \leq \max_{0\leq x, y\leq 1}\lambda_+ (x, y) +\frac 2N.
\ee

Once we established an upper and lower bound for $\Spec H$, we can obtain information about the size of possible gaps
from Corollary \ref{Corollary2}.
Thus we consider the restricted operator
\be\label{4.4}
H_N=H_N^{(0, 0)} = 2\sum_{n=0}^{N-1} (\cos 2\pi n^2\omega) e_n\otimes e_n +\sum_{n=0}^{N-1} (e_n\otimes e_{n+1}+
e_{n+1} \otimes e_n)
\ee
and denote $\lambda_1\geq \lambda_2 \geq \cdots\geq \lambda_N$ the eigenvalues of $H_N$, $\xi^{(1)}, \ldots, \xi^{(N)}$ the
corresponding normalized eigenvectors.

Our choice of $N$ here does not necessarily have to be the same as in \eqref{4.4}.
It follows from \eqref{3.5} that for any $t\in\mathbb R$
\be\label{4.5}
\text{dist\,} (t, \Spec H)\leq \min_{1\leq j\leq N} \{|t-\lambda_j|+|\xi_0^{(j)}|+|\xi^{(j)}_{N-1}|\}
\ee
with $\xi^{(j)} =\sum^{N-1}_{n=0}\xi_n^{(j)} e_n, \Vert \xi^{(j)}\Vert_2=1$.

Denote $\Gamma$ the largest gap in $\Spec H$. Thus
\be\label{4.6}
\Gamma =2 [\max_{-\sigma_+<t<\sigma_+} \text{ dist\,} (t, \Spec H)]
\ee
and by \eqref{4.5}
\be\label{4.7}
\Gamma \leq 2  \max_{-\sigma_+< t< \sigma_+} \min_{1\leq j\leq N}\{|t-\lambda_j|+|\xi_0^{(j)}|+|\xi^{(j)}_{N-1}|\}.
\ee
Instead of $H_N^{(0, 0)}$, we may as well consider $H_N^{(x, y)}$.
Hence
\be\label{4.8}
\Gamma\leq 2\max_{-\sigma_+<t<\sigma_+} \min_{x, y} \min_{1\leq j\leq N} \{|t-\lambda_j|+|\xi_0^{(x, y, j)}|+
|\xi_{N-1}^{(x, y, j)}|\}
\ee

\bigskip
\section
{Numerics}

Considering the Harper model \eqref{1.5} and the skew-shift model \eqref{1.4} with $\omega=\frac 12 (\sqrt
{5}-1)$, our purpose is to explore the finite scale behavior of the spectra and eigenvectors and compare
them for these two models.
Our numerics are based on the package MATLAB.
\medskip

\noindent
{5.1. \bf Eigenvalues and eigenvectors for the Harper model}

Set
\be\label{5.1}
H_N=2\sum^{N-1}_{n=0} (\cos 2\pi n\omega) e_n\otimes e_n+ \sum_{n=0}^{N-1} (e_n\otimes e_{n+1} + e_{n+1} \otimes e_n).
\ee
The display below shows the eigenvalue structure at $N=100, 200, 300, 400,$ and $500$.

\begin{figure}[here]
\includegraphics[scale=.50]{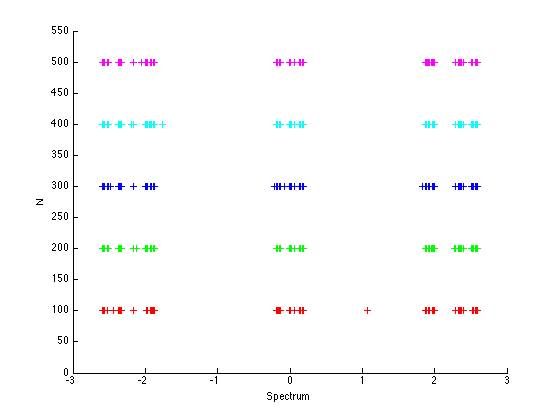}
\caption{Spectrum for the Harper Model}
\label{fig:Fig 1(a)}
\end{figure}

Note the persistency of gaps, in agreement with the (rigorously proven) Cantor structure of the spectrum of \eqref{1.1}, which
for $\lambda=1$ is in fact of zero Lebesque measure.

Next we examine the normalized eigenvectors of \eqref{5.2} for $N=200$  in different parts of the spectrum.

\pagebreak
\begin{figure}[here]
\includegraphics[scale=.50]{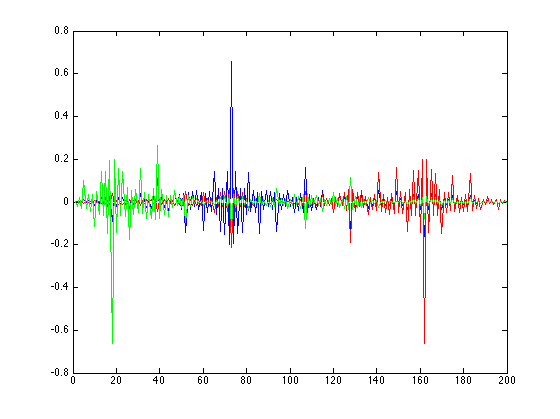}
\caption{Eigenvectors from the Left Edge of the Spectrum}
\label{fig:Fig 1(b)}
\end{figure}

\begin{figure}[here]
\includegraphics[scale=.50]{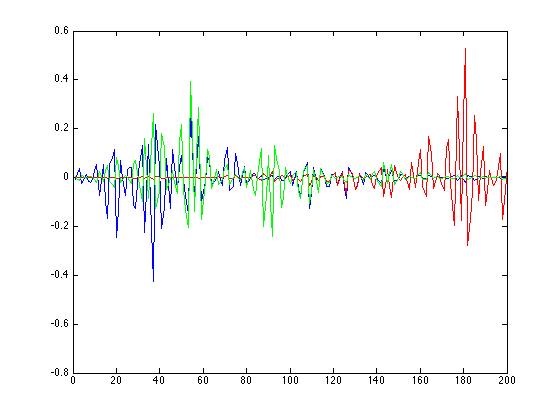}
\caption{Eigenvectors from the Center of the Spectrum}
\label{fig:Fig 1(c)}
\end{figure}

\eject
\begin{figure}[here]
\includegraphics[scale=.50]{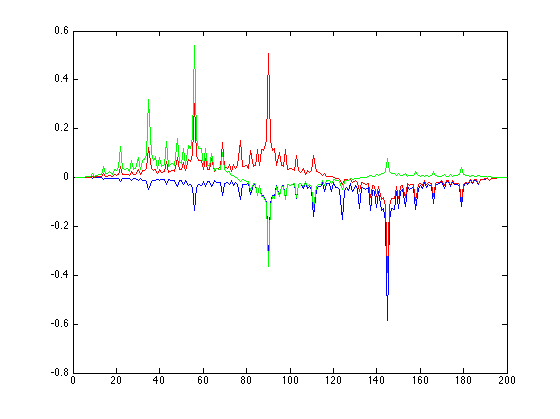}
\caption{Eigenvectors from the Right Edge of the Spectrum}
\label{fig:Fig 1(d)}
\end{figure}

\noindent
{5.2. \bf Eigenvalues and eigenvectors for the skew shift model}

Now set
\be\label{5.2}
H_N=2 \sum_{n=0}^{N-1} (\cos 2\pi n^2 \omega) e_n \otimes e_n +\sum_{n=0}^{N-1} (e_n\otimes e_{n+1}+ e_{n+1} \otimes e_n).
\ee
The eigenvalue behavior turns out to be very different as the gaps tend to close for large $N$, in agreement with the {\it conjecture} that the
spectrum of \eqref{1.4} has {\it no} gaps.

\begin{figure}[here]
\includegraphics[scale=.50]{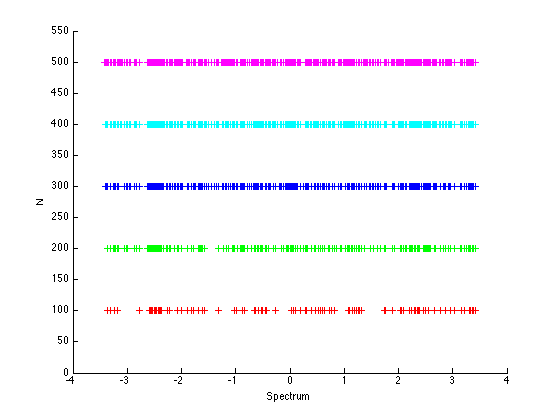}
\caption{Spectrum for Skew Shift Model}
\label{fig:Fig 2(a)}
\end{figure}
\eject

Also note that the shape of the eigenvectors is in agreement with the conjecture that \eqref {1.4} has {\it localized} states.
Below some plots for \eqref {5.2} at $N=200$ in the center and the edges of the spectrum.
\bigskip
\begin{figure}[here]
\includegraphics[scale=.50]{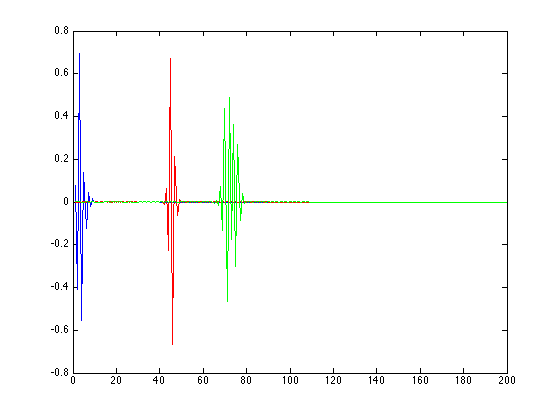}
\caption{Eigenvectors from Left Edge of Spectrum}
\label{fig:Fig 2(b)}
\end{figure}

\bigskip
\begin{figure}[here]
\includegraphics[scale=.50]{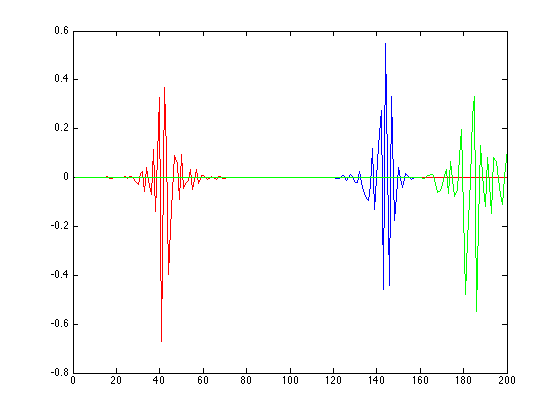}
\caption{Eigenvectors from Center of Spectrum}
\label{fig:Fig 2(c)}
\end{figure}

\begin{figure}[here]
\includegraphics[scale=.50]{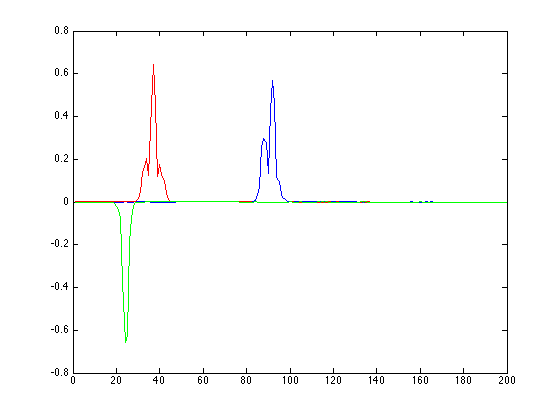}
\caption{Eigenvectors from Right Edge of Spectrum}
\label{fig:Fig 2(d)}
\end{figure}

The localized behavior is already visible at relatively low scale, as is apparent from the collective displays at $N=50, N=100$.
\bigskip
\begin{figure}[here]
\includegraphics[scale=.50]{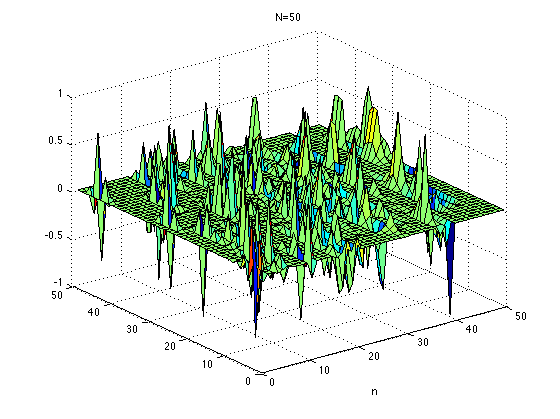}
\caption{Eigenvectors for Skew Shift Model, N=50}
\label{fig:Fig 3(a)}
\end{figure}

\begin{figure}[here]
\includegraphics[scale=.50]{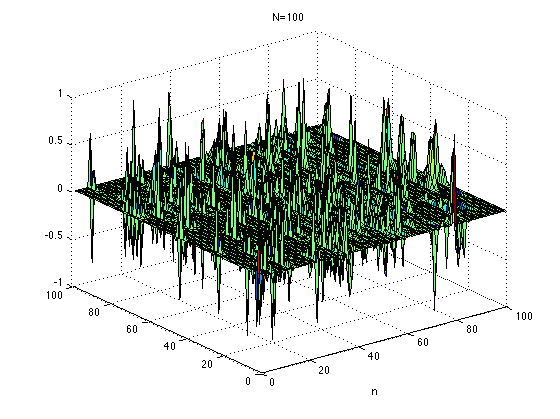}
\caption{Eigenvectors for Skew Shift Model, N=100}
\label{fig:Fig 3(b)}
\end{figure}
\bigskip
\pagebreak

\noindent
{5.3. \bf Bounding the gaps in the spectrum}

We performed some numerics pertaining to the skew shift model \eqref{4.1} as suggested by the discussion in \S4,
considering the matrices $H_N^{(x, y)}$ as defined by \eqref{4.2}.

{The first issue is a numerical evaluation of the  maximum $\sigma_+$ of the spectrum, based on \eqref{4.3}.
We obtained
\be\label{5.3}
\sigma_+ \approx 3.430
\ee

The next display graphs the function
\be\label{5.4}
\min_{x, y} \min_{1\leq j\leq N} \{|t-\lambda_j|+|\xi_0^{(x, y, j)}|+|\xi_{N-1}^{(x, y, j)}|\}
\ee
with $\{\lambda_j\}$ the eigenvalues of $H_N^{(x, y)}$ and $\{\xi^{(x, y, j)}\}$ the corresponding normalized eigenvectors.
\bigskip

\begin{figure}[here]
\includegraphics[scale=.50]{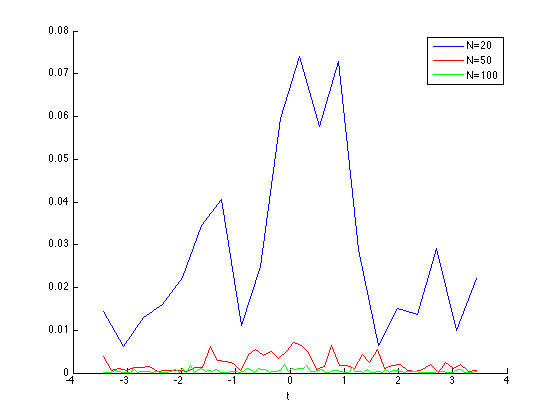}
\caption{Distance to the Spectrum for Skew Shift Model}
\label{fig:Fig 4(a)}
\end{figure}

\begin{figure}[here]
\includegraphics[scale=.50]{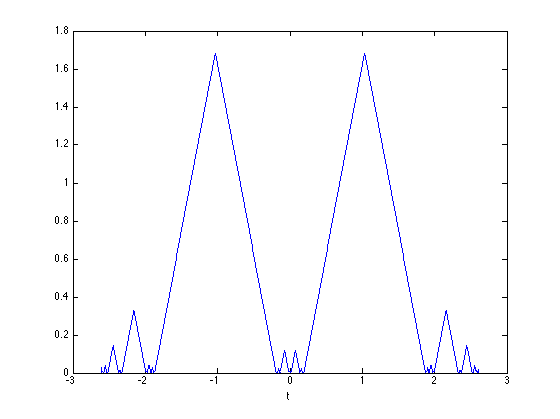}
\caption{Distance to the Spectrum for Harper Model}
\label{fig:Fig 4(b)}
\end{figure}

Based on \eqref{4.8}, an upper bound on the largest gap in the spectrum of \eqref{4.1}
\be\label{5.5}
\Gamma< 5.708 * 10^{-4}
\ee
was obtained.
Recall that $\Gamma=0$ according to the conjecture.

As a measure of comparison, we carried out these same numerics for the Harper model \eqref{1.4}.
The maximum of the spectrum now appears to be
\be\label{5.6}
\sigma_+\approx 2.5975
\ee

\noindent
The function \eqref{5.4} for the Harper model is plotted in Figure 4(b).

and a numerical estimate
\be\label{5.7}
\Gamma \approx 1.683
\ee
for the largest gap is obtained.

\bigskip
\noindent
{5.4. \bf Lyapunov exponents}

Using the notation from \S2, we computed the function
\be\label{5.8}
\frac{\log\Vert M_N(E)\Vert}{N}
\ee
where $M_N(E)$ is the transfer matrix given by \eqref{2.18} and potential
\be\label{5.9}
V_j=2\cos 2\pi j^2\omega \text { for the skew shift }
\ee
\be\label{5.10}
V_j=2\cos 2\pi j\omega \text { for the Harper model}
\ee

As was already clear from earlier calculations performed by W.~Schlag [S] around 2002, these numerics are quite unstable when $N$ becomes
reasonably large (not surprisingly so).

Figure 13 below for \eqref{5.9} at $N=20, 50, 100$ seems consistent.
\begin{figure}[here]
\includegraphics[scale=.50]{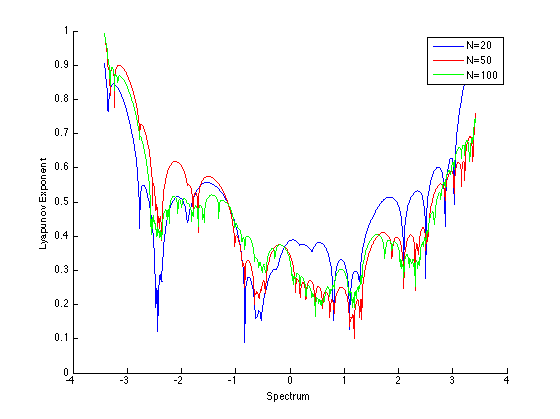}
\caption{Lyapunov Exponents for Skew Shift Model}
\label{fig:Fig 5(a)}
\end{figure}

There is also consistence with the conjecture that the Lyapunov function $L(E)$ given by \eqref{2.21} for the skew shift \eqref{1.2}
remains positive, even at small disorder $\lambda\not= 0$.
Recall that  this was only rigorously proven for $|\lambda|>1$.

Next in Figure 5(b), the corresponding display for the Harper model \eqref{5.10}, which is
in accordance with the known fact that $L(E)=0$ for $E$ in the spectrum of \eqref{1.1} when $|\lambda|\leq 1$, and the fact
that this spectrum is of zero Lebesque measure when $|\lambda|=1$.
\begin{figure}[here]
\includegraphics[scale=.50]{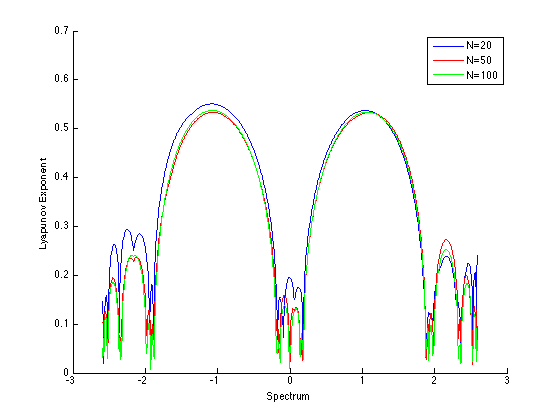}
\caption{Lyapunov Exponents for Harper Model}
\label{fig:Fig 5(b)}
\end{figure}
\bigskip
\section
{Conclusions}

The numerics carried out for the Schr\"odinger operator with skew shift potential are in convincing agreement with the
general beliefs and conjectures reviewed in the first two sections.
In particular, they indicate a spectrum with no gaps, localized eigenstates and positive Lyapunov exponents for all
energies.
This is in contrast with the Harper model where the spectrum is a Cantor set.
We have performed these numerics at different scales and obtained consistent results.
Our numerics lead moreover to a (rigorous) upperbound on the size of possible gaps (if any) for the skew shift spectrum at
the critical coupling, establishing a definitively different spectral structure compared with the Harper model.

In the case of the Harper model, our numerical findings are in accordance with the rigorously proven theoretical results, supporting the reliability of the numerical results. Few such rigorous results are available for the skew-shift counterpart, which makes the present numerical study of interest. Our main finding for the latter model is a strong indication of absence of gaps in the spectrum, as confirmed by an analysis of `truncated models' at different scales, within computational constraints.

This project offers several further research perspectives. The first is an exploration at larger scales and a finer comparison between random and pseudo-random spectra. Of particular interest is the study of the `local eigenvalue spacings' in finite models, which are known to be universal and obey Poisson statistics in the random model. One may also wish to consider $\rm{Schr\ddot{o}dinger}$ operators with frequency vectors $\omega$ other than the golden mean (as a test of consistency) and also other couplings $\lambda$. Then one could also numerically analyze $\rm{Schr\ddot{o}dinger}$ operators with different pseudo-random potentials, for instance by replacing in \eqref{1.2} the $n^2$ by $n^3$ etc. (which corresponds to higher order skew shifts) and comparing the results with the present conclusions and discussions from [B-F] and [G-F]. Finally, the problem of positivity of Lyapunov exponents remains an issue to be further studied, technically requiring a balance between scale and computational instability.

\pagebreak

\end{document}